\documentclass[draft]{beatcs}

\usepackage{amssymb}
\usepackage{amsfonts}
\usepackage{cyrillic}
\usepackage{float}
\usepackage{breakurl}

\def\sh{{\, \mathcyr{SH}\,  }}
\def\shap{{\, \mathcyr{sh} \, }}
\def\ss{{\rm ss}}

\def\ssr{{\rm ssr}}

\def\bd{{\rm bd}}
\def\bdi{{\rm bdi}}
\def\bdr{{\rm bdr}}
\def\bdir{{\rm bdir}}
\def\pss{{\rm pss}}
\def\pssr{{\rm pssr}}
\def \mod#1 #2{#1\ ({\rm mod}\ #2)}
\def\odd{{\rm odd}}
\def\even{{\rm even}}
\def\fh{{\rm fh}}
\def\lh{{\rm lh}}

\theoremstyle{plain}
\newtheorem{theorem}{Theorem}
\newtheorem{corollary}[theorem]{Corollary}
\newtheorem{lemma}[theorem]{Lemma}

\theoremstyle{definition}

\newtheorem{openproblem}[theorem]{Open Problem}

\theoremstyle{remark}
\newtheorem{remark}[theorem]{Remark}

\title{Shuffling and Unshuffling}

\author{
Dane Henshall\\
School of Computer Science\\
University of Waterloo\\
Waterloo, ON  N2L 3G1\\
Canada\\
\url{dslhensh@uwaterloo.ca}\\
\and
Narad Rampersad\\
Department of Mathematics\\
University of Li\`ege\\
Grande Traverse, 12 (Bat. B37)\\
4000 Li\`ege\\
Belgium\\
\url{narad.rampersad@gmail.com} \\
\and
Jeffrey Shallit\\
School of Computer Science\\
University of Waterloo\\
Waterloo, ON  N2L 3G1\\
Canada\\
\url{shallit@cs.uwaterloo.ca}
}

\begin{document}
\maketitle

\begin{abstract}
We consider various shuffling and unshuffling operations on languages and words,
and examine their closure properties.  Although the main goal is to provide
some good and novel exercises and examples for undergraduate formal language
theory classes, we also
provide some new results and some 
open problems.
\end{abstract}

\section{Introduction}

    Two kinds of shuffles are commonly studied:  perfect shuffle and
ordinary shuffle.

    For two words $x = a_1 a_2 \cdots a_n$, $y = b_1 b_2 \cdots b_n$ of
the same length, we define their {\em perfect shuffle}
$x \shap y = a_1 b_1 a_2 b_2 \cdots a_n b_n$.  For example,
${\tt term} \shap {\tt hoes} = {\tt theorems}$.  
Note that $x \shap y$ need not equal $y \shap x$.  
This definition is
extended to languages as follows:
$$ L_1 \shap L_2 = \bigcup_{{x \in L_1,\ y\in L_2}\atop
	{|x| = |y|}} \lbrace x \shap y \rbrace .$$
If $x^R$ denotes the reverse of $x$, then note that 
$(x \shap y)^R = y^R \shap x^R$.

It is sometimes useful to allow $|y| = |x|+1$, where
$x = a_1 \cdots a_n$, $y = b_1 \cdots b_{n+1}$, in which case we define
$x \shap y = a_1 b_1 \cdots a_n b_n b_{n+1}$.

     The {\em ordinary shuffle} $x \sh y$ of two words is a finite
{\em set}, the set
of words obtainable from merging the words $x$ and $y$ from left to
right, but choosing the next symbol arbitrarily from $x$ or $y$.
More formally,
\begin{multline*}
x \sh y = \lbrace z \ : \ z = x_1 y_1 x_2 y_2 \cdots x_n y_n 
\text{ for some $n \geq 1$ and } \\
\text{words $x_1, \ldots, x_n, y_1, \ldots, y_n$
such that } x = x_1 \cdots x_n \text{ and } y = y_1 \cdots y_n \rbrace.
\end{multline*}
This definition is symmetric, and $x \sh y = y \sh x$.  The definition
is extended to languages as follows:
$$ L_1 \sh L_2 = \bigcup_{x \in L_1,\ y\in L_2} (x \sh y).$$

(As a mnemonic, the symbol $\sh$ is larger than $\shap$ in size, and
similarly $\sh$ generally produces a set larger in cardinality than
$\shap$.)

As is well-known, the shuffle (resp., perfect shuffle) of two regular
languages is regular, and the shuffle (resp., perfect 
shuffle) of a context-free language
with a regular language is context-free.  Perhaps the easiest way
to see all these results
is by using morphisms and inverse morphisms, and relying on the 
known closure properties of these transformations, as follows:

If $L_1, L_2 \subseteq \Sigma^*$, create a new alphabet $\Sigma'$ by
putting primes on all the letters of $\Sigma$.  Define
$h_1 (a) = h_2(a') = a$ and $h_1 (a') = h_2(a) = \epsilon$ for $a \in \Sigma$.
Define $h(a) = h(a') = a$ for $a \in \Sigma$.
Then 
$$L_1 \sh L_2 = h( h_1^{-1} (L_1) \ \cap \ h_2^{-1} (L_2) ).$$

In a similar way, 
$$ L_1 \shap L_2 = h( h_1^{-1}(L_1) \ \cap \ h_2^{-1} (L_2) \ \cap
	(\Sigma \Sigma')^* ).$$

However, the shuffle
(resp., perfect shuffle) of two context-free languages need not be
context-free.  For example, if $L_1 = \lbrace a^m b^m \ : \ m \geq 1 \rbrace$
and $L_2 = \lbrace c^n d^n \ : \ n \geq 1 \rbrace$, then
$L := L_1 \sh L_2$ is not a CFL.  If it were, then
$L \ \cap \ a^+ c^+ b^+ d^+ = \lbrace a^m c^n b^m d^n \ : \ m, n \geq 1
\rbrace$ would be a CFL, which it isn't (via the pumping lemma).

Similarly, if $L_3 = \lbrace a^m b^{2m} \ : \ m \geq 1 \rbrace$
and $L_4 = \lbrace a^{2n} b^n \ : \ n \geq 1 \rbrace$, then
$L_3 \shap L_4 = \lbrace a^{2n} (ba)^n b^{2n} \ : \ n \geq 1 \rbrace$,
which is clearly  not a CFL.

For these, and other facts, see \cite{Berstel:1979}.

\section{Self-shuffles}

Instead of shuffling languages together, we can take a language and
shuffle (resp., perfect shuffle) each word with itself.
Another variation is to shuffle each
word with its reverse.  This gives four different transformations
on languages, which we call self-shuffles:

\begin{eqnarray*}
\ss(L) &=& \bigcup_{x \in L} \{ x \sh x \}  \\
\pss(L) &=& \bigcup_{x \in L} x \shap x \\
\ssr(L) &=& \bigcup_{x \in L} \{ x \sh x^R \}  \\
\pssr(L) &=& \bigcup_{x \in L} x \shap x^R  .
\end{eqnarray*}

We would like to understand how these transformations affect regular
and context-free languages.  We obtain some results, but other questions
are still open.

\begin{theorem}
If $L$ is regular, then
$\ss(L)$ need not be context-free.
\end{theorem}

\begin{proof}  
We show that $\ss(\lbrace 0,1 \rbrace^*)$ is not a CFL.  Suppose it is,
and consider 
$L' = \ss(\{0,1\}^{*}) \cap R$, where
$R=\{ 01^{a}0^{b+1}1^{c+1}0^{d}1 \ : \ a,b,c,d \geq 1 \}$.
Since $R$ is regular, it suffices to 
show that $L'$ is not context-free. 

Now consider an arbitrary word $w \in L'$. 
Then $w=01^{a}0^{b+1}1^{c+1}0^{d}1$ for some $a,b,c,d \geq 1$, and there exists a $y \in \{0,1\}^{*}$ such that $w \in y \sh y$.
The structure of $w$ allows us to determine $y$.
Let $y_1$ and $y_2$ be copies of $y$ such that $w \in y_1 \sh y_2$,
and the first letter of $w$ is taken from $y_1$.

The first symbol of $y$ is evidently
$0$. It follows that the prefix $01^a$ of $w$ is taken entirely from $y_1$, 
since the $0$ is taken from $y_1$ by definition and the first symbol of $y_2$ is 0. Therefore $01^a$ is a prefix of $y_1$.

It follows that $y_2$ also contains $01^a$ as a prefix,
and since $a \geq 1$ this is only possible if the first $0$ of $y_2$
is located in the $0^{b+1}$ block of $w$.
Otherwise, $y_2$ would be a subsequence of $0^{d}1$ and $y_1$ would have $01^{a}0^{b+1}1^{c+1}$ as a prefix (implying that $y_1 \neq y_2$).
Furthermore, the second symbol of $y_2$ being $1$ implies that exactly one of the 0's in the $0^{b+1}$ block is from $y_2$. 
Thus the rest are from $y_1$ and $01^{a}0^{b}$ is a prefix of $y_1$.

Note that $y_1$ and $y_2$ both end in $1$, and $w$ ends in $0^{d}1$.
By the same logic as before, we can conclude that $0^{d}1$ is
a suffix of exactly one of them,
and that the other ends in the $1^{c+1}$ block.
Thus $y_2$ contains $0^{d}1$ as a suffix
and $y_1$ ends in the $1^{c+1}$ block (otherwise, $y_1 \neq y_2$).

Finally, since the second last symbol of $y_1$ is $0$
and $y_1$ ends in the $1^{c+1}$ block, we can conclude that $y_1$ contains 
exactly one $1$ from the $1^{c+1}$ block and that $y_1=01^{a}0^{b}1$.
Unshuffling $y_1$ from $w$ yields $y_2 = 01^{c}0^{d}1$. 

Recall that $y_1=y_2$. So,
\[
 y_1=01^{a}0^{b}1 = 01^{c}0^{d}1 = y_2
\]
and since $a,b,c,d \geq 1$ we know that
\[
 a = c \text{ \indent and \indent } b = d.
\]
\\
If $w \in L'$ then
\begin{align*}
 w & =01^{a}0^{b+1}1^{c+1}0^{d}1 \\
   & = 01^{a}0^{d+1}1^{a+1}0^{d}1 \\
   & =01^{a}0^{d}(01)1^{a}0^{d}1 .
\end{align*}

Since $w$ was arbitrary, we have
\begin{align*}
 L' & = \{ 01^{a}0^{b+1}1^{c+1}0^{d}1 : a=c, b=d, \text{ and } a,d \geq 1 \} \\
   & = \{ 01^{n}0^{m}(01)1^{n}0^{m}1 : m,n \geq 1 \} ,
\end{align*}
which is clearly not a CFL, using the pumping lemma.
\end{proof}

\begin{remark}
In a previous version of this paper, proving that
$\ss(\lbrace 0,1 \rbrace^*)$ is not context-free
was listed as an open problem.  After this was solved by D. Henshall,
a solution was
given by Georg Zetzsche independently.
\end{remark}

Similarly, we can show

\begin{theorem}
$L = \bigcup_{w \in \lbrace 0, 1 \rbrace^*}  ( w \sh w \sh w ) $
is not context-free.
\label{three}
\end{theorem}

\begin{proof}  
We use Ogden's lemma.  Consider
$$  L =  \lbrace w \sh w \sh w \ : \ w \in \lbrace 0,1 \rbrace^* \rbrace
\cap 0^*10^*10^*1.$$

Pick $s = 0^n 1 0^n 1 0^n 1$ in $L$ to pump.  Write $s = uvxyz$ and mark the
middle block of 0's.  If $v$ begins in the middle block of 0's, then
pump up to obtain $s' = 0^n 1 0^j 1 0^k 1$, where $n < j$ and $n \leq k$.  We
can't have $s' \in w \sh w \sh w $ because the first $w$ (the one ending at
the first 1) is too short.  If $v$ begins in the first block of 0's,
then $y$ occurs in the middle block, so now pump down to obtain $s' = 0^i
1 0^j 1 0^n 1$, where $i \leq n$ and $j < n$.  Again, we can't have $s'
\in w \sh w \sh w$, because the third $w$ (the one ending at the third 1) must
contain all of the $0$'s immediately preceding the final 1, and hence is
too long.
\end{proof}

Clearly $\ss(\lbrace 0,1 \rbrace^*)$ is in NP, since given a
word $w$ we can guess $x$ and check that $w \in x \sh x$.  
However, we do not know whether we can solve membership for
$\ss(\lbrace 0,1 \rbrace^*)$ in polynomial time.  This question 
is apparently originally due to Jeff Erickson \cite{Erickson:2010},
and we learned about it from Erik Demaine.

\begin{openproblem}
Is $\ss(\lbrace 0,1 \rbrace^*)$ in P?
\label{ope}
\end{openproblem}

We mention a few related problems.  
Mansfield \cite{Mansfield:1982} showed that, given words $w, x, y$,
one can decide in polynomial
time if $w \in x \sh y$.    Later, the same author 
\cite{Mansfield:1983} and, independently, Warmuth and Haussler
\cite{Warmuth&Haussler:1984} showed that, given words
$w, x_1, x_2, \ldots, x_n$, deciding if
$w \in x_1 \sh x_2 \sh \cdots \sh x_n$ is NP-complete.  However,
the decision problem implied by
Open Problem~\ref{ope} asks something different:  given
$w$, does there exist $x$ such that $w \in x \sh x$?

\begin{openproblem}
Determine a simple closed form for
$$ a_k(n) := \left| \ \bigcup_{x \in \lbrace 0, 1, \ldots, k-1 \rbrace^n} 
(x \sh x) \ \right|.$$
\end{openproblem}

The first few terms are given as follows:
\begin{figure}[H]
\begin{center}
\begin{tabular}{|c|c|c|c|c|c|c|c|c|c|c|}
\hline
$n$ & 0 & 1 & 2 & 3 & 4 & 5 & 6 & 7 & 8 & 9 \\
\hline
$a_2(n)$ & 1 & 2 & 6 & 22 & 82 & 320 & 1268 & 5102 & 20632 & 83972\\
\hline
$a_3(n)$ & 1 & 3 & 15 & 93 & 621 & 4425 & 32703 & 248901  & & \\
\hline
$a_4(n)$ & 1 & 4 & 28 & 244 & 2332 & 23848 & 254416 & & & \\
\hline
$a_5 (n)$ & 1 & 5 & 45 & 505 & 6265 & 83225 & & & &  \\
\hline
$a_6 (n)$ & 1 & 6 & 66 & 906 & 13806 & 225336 & & & & \\
\hline
\end{tabular}
\end{center}
\end{figure}

Clearly
$a_i (0) = 1$, $a_i(1) = i$, and $a_i (2) = 2i^2-i$.  Empirically we have
$a_i (3) = 5i^3-5i^2+i$, 
$a_i (4) = 14i^4 - 21i^3 + 5i^2+ 3i$, and
$a_i (5) = 42 i^5 - 84i^4 + 32 i^3 + 21i^2 - 10i$.   This suggests that
$a_i (n) = {{{2n} \choose n}\over {n+1}} i^n - {{2n-1} \choose {n+1}}i^{n-1} +
	O(i^{n-2})$, but we do not have a proof.

\section{Perfect self-shuffle}

We can consider the same question for perfect shuffle.  
We define 
$$\pss(L) = \bigcup_{x \in L} \lbrace x \shap x \rbrace.$$

\begin{theorem}
Both the class of regular languages and the class of context-free
languages are closed under $\pss$.
\end{theorem}

\begin{proof}
Use the fact that $\pss(L) = h(L)$, where $h$ is the morphism mapping
$a \rightarrow aa$ for each letter $a$.
\end{proof}

\section{Self-shuffle with reverse}

We now characterize those words $y$ that can be written as a shuffle of
a word with its reverse; that is, as a member of the set $x \sh x^R$.

An {\it abelian square} is a word of the form $x x'$ where $x'$ is a
permutation of $x$.

\begin{theorem}
(a) If there exists $x$ such that $y \in x \sh x^R$, 
then $y$ is an abelian square.

(b) If $y$ is a binary abelian square, then there exists
$x$ such that $y \in x \sh x^R$.
\end{theorem}

We introduce the following notation:  if $w = a_1 a_2 \cdots a_n$, then
by $w[i..j]$ we mean the factor $a_i a_{i+1} \cdots a_j$.

\begin{proof}
(a)
If $y$ is the shuffle of $x$ with its reverse,
then the first half of $y$ must contain some prefix of $x$, say $x[1..k]$.
Then the second half of $y$ must contain the remaining suffix of $x$, say
$x[k+1..n]$.  Then the second half of $y$ must contain,
in the remaining positions, some prefix of $x$, reversed.  But by counting we
see that this prefix must be $x[1..k]$.
So the first half of $y$ must contain the remaining symbols of $x$, reversed.
This shows that the
first half of $y$ is just $x[1..k]$ shuffled with $x[k+1..n]^R$, and
the second half of $y$ is just $x[k+1..n]$ shuffled with $x[1..k]^R$.

So the second half of $y$ is a permutation of the first half of $y$.

(b) 
It remains to see that every binary abelian square can be obtained in this way.

To see this, note that if $x$ contains $j$ 0's and $n-j$ 1's, then we can get
$y$ by shuffling $0^j 1^{n-j}$ with its reverse.  We get the $0$'s in $x$ by
choosing them from $0^j 1^{n-j}$, and we get the $1$'s in $x$ by choosing
them from $(0^j 1^{n-j})^R$.
\end{proof}

\begin{remark}
The word $012012$ is an example of a ternary abelian square that cannot
be written as an element of $w \sh w^R$ for any word $w$.
\end{remark}

\begin{remark}
The preceding proof gives another proof of the
classic identity
$${{2n} \choose n} = {n \choose 0}^2 + \cdots +  {n \choose n}^2 .$$
To see this, we use the following bijections:  the binary
words of length $2n$ having exactly $n$ $0$'s (and hence $n$ $1$'s) are in
one-one correspondence with the abelian squares of length $2n$, as follows:
take such a word and complement the last $n$ bits.   Thus there are
${2n} \choose n$ binary abelian squares of length $2n$.

On the other hand,
there are ${n \choose i}^2$
words that are abelian squares and have a first and last half,
each with $i$ 0's.
Summing this from $i = 0$ to $n$ gives the result.
\end{remark}

\begin{corollary}
The language
$$ \ssr(\lbrace 0,1 \rbrace^*) = \bigcup_{x \in \lbrace 0,1 \rbrace^*} x \sh x^R $$
is not a CFL, but is in P.
\end{corollary}

\begin{proof}
From above, intersecting $ \ssr(\lbrace 0,1 \rbrace^*)$ 
with $0^+ 1^+ 0^+ 1^+$
gives 
$$\lbrace 0^m 1^n 0^{m+2k} 1^n \ : \ m, n \geq 1 \text{ and } k \geq 0 \rbrace
\ \cup \ \lbrace 0^m 1^{n+2k} 0^m 1^n \ : \ m,n \geq 1 \text{ and }  k \geq 0
\rbrace.$$
Now the pumping lemma applied to $z = 0^n 1^n 0^n 1^n$ shows this is not
a CFL.

Since we can easily test if a string is an abelian square by counting
the number of $0$'s in the first half, and comparing it to the number of
$0$'s in the second half, it follows that
$\ssr(\lbrace 0,1 \rbrace^*)$ is in P.
\end{proof}

As before, we can define 
$$ b_k(n) := \left| \ \bigcup_{x \in \lbrace 0, 1, \ldots, k-1 \rbrace^n}
(x \sh x^R) \ \right|.$$
For $k = 2$, our results above explain $b_k (n)$, but
we do not know a closed form for larger $k$.

The first few terms are given as follows:
\begin{figure}[H]
\begin{center}
\begin{tabular}{|c|c|c|c|c|c|c|c|c|c|c|}
\hline
$n$ & 0 & 1 & 2 & 3 & 4 & 5 & 6 & 7 & 8 & 9 \\
\hline
$b_2(n)$ & 1 & 2 & 6 & 20 & 70 & 252 & 924 & 3432 & 12870 & 48620\\
\hline
$b_3(n)$ & 1 & 3 & 15 & 87 & 549 & 3657 & 25317 & 180459 & & \\
\hline
$b_4(n)$ & 1 & 4 & 28  & 232 & 2116 & 20560 & 208912 & & & \\
\hline
$b_5 (n)$ & 1 & 5 & 45 & 485 & 5785 & 73785 & & & &  \\
\hline
$b_6 (n)$ & 1 & 6 & 66 & 876 & 12906 & 203676 & & & & \\
\hline
\end{tabular}
\end{center}
\end{figure}

Clearly
$b_i (0) = 1$, $b_i (1) = i$, and $b_i(2) = 2i^2 -i $.
Empirically, we have
$b_i (3) = 5i^3-6i^2+2i$, $b_i (4) = 14i^4 - 27i^3 + 17i^2 - 3i$, and
$b_i (5) = 42i^5 -110 i^4 + 94 i^3 -17 i^2 - 8i$.  This suggests that
$b_i (n) = {{{2n}\choose n}\over {n+1}} i^n - \left( {{2n-1} \choose n-1} -2^{n-1}
\right ) i^{n-1} + O(i^{n-2})$, but we do not have a proof.

\section{Perfect self-shuffle with reverse}

We now consider the operation $w \rightarrow w \shap w^R$ applied to
languages.  Recall that $\pssr(L) = \bigcup_{x \in L} \lbrace
x \shap x^R \rbrace$.

\begin{theorem}
If $L$ is regular then $\pssr(L)$ is not necessarily regular.
\end{theorem}

\begin{proof}
Let $L = 0^+ 1 0^+$.  Then $\pssr(L) \ \cap \ 0^+ 1 1 0^+ =
\lbrace 0^n 1 1 0^n \ : \ n \geq 2 \rbrace$, which is clearly
not regular.
\end{proof}

\begin{theorem}
If $L$ is context-free then $\pssr(L)$ is not necessarily context-free.
\end{theorem}

\begin{proof}
Let $L = \lbrace 0^m 1^m 2^n 3^n \ :  \ m, n\geq 1 \rbrace$.  
Then $\pssr(L) \ \cap \ (03)^+ (12)^+ (21)^+ (30)^+ =
\lbrace (03)^n (12)^n (21)^n (30)^n \ : \ n \geq 1 \rbrace$,
and this language is easily seen to be non-context-free.
\end{proof}

\begin{theorem}
If $L$ is regular then $\pssr(L)$ is necessarily context-free.
\label{pssr}
\end{theorem}

We defer the proof of Theorem~\ref{pssr} until Section~\ref{idwr}
below.

\section{Unshuffling}

Given a finite word $w = a_1 a_2 \cdots a_n$ we can decimate it into its
odd- and even-indexed parts, as follows:
\begin{eqnarray*}
\odd(w) &=& a_1 a_3 \cdots a_{n-((n+1) \bmod 2)} \\
\even(w) &=& a_2 a_4 \cdots a_{n-(n \bmod 2)}
\end{eqnarray*}
Similarly, given $w = a_1 a_2 \cdots a_n$ we can extract its first and
last halves, as follows:
\begin{eqnarray*}
\fh(w) &=& a_1 a_2 \cdots a_{\lfloor n/2 \rfloor} \\
\lh(w) &=& a_{\lfloor n/2 \rfloor + 1} \cdots a_n 
\end{eqnarray*}

We now turn our attention to four ``unshuffling'' operations:
\begin{eqnarray*}
\bd(w) &=& \odd(w) \even(w) \\
\bdr(w) &=& \odd(w) \even(w)^R \\
\bdi(w) &=& \fh(w) \shap \lh(w) \\
\bdir(w) &=& \fh(w) \shap \lh(w)^R
\end{eqnarray*}

\subsection{Binary decimation}

We first consider a kind of binary decimation, which forms a sort of inverse
to perfect shuffle.

Given a word $w = a_1 a_2 \cdots a_{2n}$ of even length, note that
$$\bd (w) =  a_1 a_3 \cdots a_{2n-1} a_2 a_4 \cdots a_{2n}$$
is formed
by ``unshuffling'' the word into its odd- and even-indexed letters.
For example, the French word {\tt maigre} becomes the word {\tt mirage}
under this operation.  

\begin{theorem}
Neither the class of regular languages nor the class of context-free
languages is closed under 
$\bd$.
\end{theorem}

\begin{proof}
Consider the regular (and context-free) language
$L = (00+11)^+$.  Then $\bd(L) = \lbrace w w \ : \ w \in \lbrace
0, 1 \rbrace^+ \rbrace$, which is well-known to be non-context-free.
\end{proof}

\subsection{Binary decimation with reverse}

We now consider the operation $\bdr$, which is a kind of binary
decimation with reverse.
Note that
$$\bdr(a_1 a_2 \cdots a_{2n}) =
a_1 a_3 \cdots a_{2n-1} a_{2n} \cdots a_4 a_2.$$
For example, $\bdr({\tt friend}) = {\tt finder}$
and $\bdr({\tt perverse}) = {\tt preserve}$.

\begin{theorem}
The class of regular languages is not closed under $\bdr$.
\end{theorem}

\begin{proof}
Let $L = (00)^+ 11$.  Then $\bdr(L) = \lbrace 0^n 11 0^n \ : \ n
\geq 1 \rbrace$, which is not regular. 
\end{proof} 

\begin{theorem}
The class of context-free languages is not closed under $\bdr$.
\end{theorem}

\begin{proof}
Consider $L = \lbrace (03)^n (12)^n \ : \ n \geq 1 \rbrace$.
Then $\bdr(L) = \lbrace 0^n 1^n 2^n 3^n  \ : \ n \geq 1 \rbrace$, 
which is not context-free.
\end{proof}

\begin{theorem}
If $L$ is regular, then $\bdr(L)$ is context-free.
\end{theorem}

\begin{proof}
We show how to accept words of $\bdr(L)$ of even length; words of odd length
can be treated similarly. 

On input $w = b_1 b_2 \cdots b_{2n}$, a PDA can guess
$x = a_1 a_2 \cdots a_{2n}$
in parallel with the elements of the input.  
At each stage the PDA compares $a_i$ to $b_{(i+1)/2}$ if $i$ is odd;
and otherwise it pushes $a_i$ onto the stack (if $i$ is even).  At some
point the PDA nondeterministically guesses that it has seen $a_{2n}$ and
pushed it on the stack; it now pops the stack (which is holding $a_{2n}
\cdots a_4 a_2$) and compares the stack contents to the rest of the
input $w$.

The PDA accepts if $x \in L$ and the symbols matched as described.
\end{proof}

\subsection{Inverse decimation}

We now consider a kind of inverse decimation, which shuffles the
first and last halves of a word.  

Note that if $w = a_1 \cdots a_{2n}$
is of even length, then 
$$\bdi( w) =
a_1 a_{n+1} a_2 a_{n+2} \cdots a_n a_{2n}.$$
Further, $\bdi(\bd(w)) = \bd(\bdi(w))$ for $w$ of even length.

\begin{theorem}
If $L$ is regular then so is $\bdi (L)$.
\end{theorem}

\begin{proof}
On input $x$ we simulate the DFA for $L$ on the odd-indexed letters
of $x$, starting from $q_0$, and we simulate a second copy of
the DFA for $L$ on the even-indexed letters,  starting
at some guessed state $q$.  Finally, we check to see that our guess
of $q$ was correct.
\end{proof}

\begin{theorem}
The class of context-free languages is not closed under
$\bdi$.
\end{theorem}

\begin{proof}
Let $L = \lbrace 0^{m} 1^{m} 2^{2n} 3^{4n} \ : \ m,n \geq 1 \rbrace$.
It is easy to see that

$$\bdi(L) = \begin{cases}
(01)^{m-3n} (02)^{2n} (03)^{n} (13)^{3n}, & \text{ if $m \geq 3n$}; \\
(02)^{m-n} (03)^n (13)^m (23)^{3n-m}, & \text {if $n \leq m \leq 3n$}; \\
(03)^m (13)^m (23)^{2n} (33)^{n-m}, & \text{if $m \leq n$}.
\end{cases}
$$

Consider $L' := \bdi(L) \ \cap \ (03)^+ (13)^+ (23)^+$.
From the above we have $L' = \lbrace (03)^n (13)^n (23)^{2n} \ : \
n \geq 1 \rbrace$, which is evidently not context-free.
\end{proof}

\subsection{Inverse decimation with reverse}
\label{idwr}

Note that if $w = a_1 \cdots a_{2n}$
is of even length, then $\bdir( w) =
a_1 a_{2n} a_2 a_{2n-1} \cdots a_n a_{n+1}$. 
If $w = a_1 \cdots a_{2n+1}$ is of odd length,
we define 
$$\bdir(w) = a_1 a_{2n+1} a_2 a_{2n} \cdots a_n a_{n+2} a_{n+1}.$$

\begin{theorem}
If $L$ is regular then so is $\bdir (L)$.
\label{bdir}
\end{theorem}

\begin{proof}
On input $x$ we simulate the DFA $M$ for $L$ on the odd-indexed letters
of $x$, starting from $q_0$.  We also create an NFA $M'$ accepting $L^R$ in
the usual manner, by reversing the transitions of $M$, and making
the start state the set of final states of $M$,
and we simulate $M'$ on the even-indexed letters of $x$.
Finally, we check to see that we meet in the middle.
\end{proof}

\begin{theorem}
The class of context-free languages is not closed under $\bdir$.
\label{bdircf}
\end{theorem}

\begin{proof}
Consider $L = \lbrace 0^{2m} 1^{4m} 2^n 3^n \ : \ m, n \geq 1 \rbrace$.
Then $L$ is a CFL, and it is easy to verify that
$$
\bdir( 0^{2m} 1^{4m} 2^n 3^n ) =
\begin{cases}
(03)^n (02)^n (01)^{2m-2n} (11)^{m+n}, \text{if $m \geq n$;} \\
(03)^n (02)^{2m-n} (12)^{2n-2m} (11)^{3m-n}, \text{if $m \leq n \leq 2m$;} \\
(03)^{2m} (13)^{n-2m} (12)^n (11)^{3m-n}, \text{if $2m \leq n \leq 3m$;} \\
(03)^{2m} (13)^{n-2m} (12)^{6m-n} (22)^{n-3m}, \text{if $3m \leq n \leq 6m$;} \\
(03)^{2m} (13)^{4m} (23)^{n-6m} (22)^{3m}, \text{if $n \geq 6m$.}
\end{cases}
$$
Assume $\bdir( L)$ is a CFL. Then
$L' := \bdir(L) \ \cap \ (03)^+ (13)^+ (22)^+ $ is a CFL, and from
above we have $L' = \lbrace (03)^{2m} (13)^{4m} (22)^{3m} \ : \ m \geq 1
\rbrace$, which is not a CFL.
\end{proof}

As Georg Zetzsche has kindly pointed out to us, the operation
$\bdir$ was studied previously by Jantzen and Petersen
\cite{Jantzen&Petersen:1994}; they called it ``twist''.
They proved our Theorems~\ref{bdir} and \ref{bdircf}.

We now return to the proof of Theorem~\ref{pssr}, which was postponed until
now.  We need two lemmas:

\begin{lemma}
Suppose $L$ is a regular language.  Then
$L' = \lbrace w w^R \ : \ w \in L \rbrace$ is a CFL.
\label{ww}
\end{lemma}

\begin{proof}
On input $x$, a PDA can guess $w$ and verify it is in $L$,
while pushing it on the
stack.  Nondeterministically it then guesses it is at the end of $w$ and pops
the stack, comparing to the input.
\end{proof}

\begin{lemma}
For all words $w$ we have $w \shap w^R = \bdir(w) \, \bdir(w)^R$.
\label{bdir2}
\end{lemma}

\begin{proof}
If $w$ is of even length then
\begin{eqnarray*}
w \shap w^R & = & (\fh(w) \lh(w)) \shap (\fh(w) \lh(w))^R \\
&=& (\fh(w) \lh(w)) \shap (\lh(w)^R \fh(w)^R) \\
&=& (\fh(w) \shap \lh(w)^R) (\lh(w) \shap \fh(w)^R) \\
&=& \bdir(w) \bdir(w)^R.
\end{eqnarray*}
A similar proof works for $w$ of odd length.
\end{proof}

We can now prove Theorem~\ref{pssr}.

\begin{proof}
From Lemma~\ref{bdir2} we have
$$\pssr(L) = \bigcup_{x \in L} x \shap x^R
= \bigcup_{x \in L} \bdir(x) \, \bdir(x)^R  = 
\bigcup_{x \in \bdir(L)} x x^R.$$
If $L$ is regular, then $\bdir(L)$ is regular, by
Theorem~\ref{bdir}.    Then, from Lemma~\ref{ww}, it follows
that $\pssr(L)$ is a CFL.
\end{proof}

%
%
%
%
%
%
%
%
%

\section{Acknowledgment}

We are grateful to Georg Zetzsche for his remarks.


\begin{thebibliography}{1}

\bibitem{Berstel:1979}
\newblock J. Berstel.
\newblock {\it Transductions and Context-Free Languages}.
\newblock Teubner, 1979.

\bibitem{Erickson:2010}
\newblock J. Erickson.
\newblock How hard is unshuffling a string? \\
\newblock \url{http://cstheory.stackexchange.com/questions/34/how-hard-is-unshuffling-a-string}, August 16 2010.

\bibitem{Jantzen&Petersen:1994}
\newblock M. Jantzen and H. Petersen.
\newblock Cancellation in context-free languages: enrichment by reduction.
\newblock {\it Theoret. Comput. Sci.} {\bf 127} (1994), 149--170.

\bibitem{Mansfield:1982}
\newblock A. Mansfield.
\newblock An algorithm for a merge recognition problem.
\newblock {\it Disc. Appl. Math.} {\bf 4} (1982), 193--197.

\bibitem{Mansfield:1983}
\newblock A. Mansfield.
\newblock On the computational complexity of a merge recognition problem.
\newblock {\it Disc. Appl. Math.} {\bf 5} (1983), 119--122.

\bibitem{Warmuth&Haussler:1984}
\newblock M. K. Warmuth and D. Haussler.
\newblock On the complexity of iterated shuffle.
\newblock {\it J. Comput. Sys. Sci.} {\bf 28} (1984), 345--358.

\end{thebibliography}
\end{document}